\documentclass[journal]{IEEEtran}

\IEEEoverridecommandlockouts

\usepackage{mathrsfs}
\usepackage{amsxtra}
\usepackage{latexsym}
\usepackage{graphicx}
\usepackage{epsfig}
\usepackage{subfigure}
\usepackage{array}
\usepackage{amsmath}
\usepackage{amssymb}
\usepackage{color, soul}
\usepackage{amsthm}
\usepackage{enumerate}
\usepackage{algpseudocode}
\usepackage{algorithm}
\usepackage{bm}
\usepackage{amssymb}
\usepackage{balance}

\theoremstyle{plain} 

\newtheorem{prop}{Proposition}
\theoremstyle{definition}

\newtheorem{remark}{Remark}

\def\bal#1\eal{\begin{align}#1\end{align}}

\newcommand{\bH}{{\bf H}}

\newcommand{\bY}{{\bf Y}}

\newcommand{\bI}{{\bf I}}

\newcommand{\bV}{{\bf V}}
\newcommand{\bg}{{\bf g}}

\newcommand{\bA}{{\bf A}}

\newcommand{\bQ}{{\bf Q}}

\newcommand{\bS}{{\bf S}}
\newcommand{\bq}{{\bf q}}

\newcommand{\bv}{{\bf v}}
\newcommand{\bX}{{\bf X}}

\newcommand{\bu}{{\bf u}}

\newcommand{\bx}{{\bf x}}

\newcommand{\ba}{{\bf a}}

\newcommand{\bn}{{\bf n}}

\newcommand{\bo}{{\bf 0}}

\newcommand{\bw}{{\bf w}}
\newcommand{\bh}{{\bf h}}

\newcommand{\bp} {\begin{proof}}
\newcommand{\ep} {\end{proof}}

\newcommand{{\Rb}} {\right)}

\newcommand{{\Rf}} {\right\}}
\newcommand{{\diag}} {\mathrm{diag}}

\begin{document}

\title{\LARGE{  Active Reconfigurable Intelligent Surface Aided Secure Transmission }}

\author{Limeng Dong, Hui-Ming Wang \emph{Senior Member, IEEE}, and Jiale Bai

\thanks{This work was in part finished when Limeng Dong was working in the  School of Information and Communications Engineering, and also with the Ministry of Education Key Laboratory for Intelligent Networks and Network Security, Xi'an Jiaotong University, Xi'an 710049, China. Limeng Dong is  with the School of Electronics and Information, Northwestern Polytechnical University, Xi'an 710129, China, email: dlm$\_$nwpu@hotmail.com.}

\thanks{Hui-Ming Wang and Jiale Bai are with the  School of Information and Communications Engineering, and also with the Ministry of Education Key Laboratory for Intelligent Networks and Network Security, Xi'an Jiaotong University, Xi'an 710049, China (e-mail: xjbswhm@gmail.com; bjl19970954@stu.xjtu.edu.cn.)}

}

\maketitle

\begin{abstract}

Reconfigurable Intelligent Surface (RIS) draws great attentions in academic and industry due to its passive and low power consumption nature, and has currently been used in physical layer security to enhance the secure transmission. However, due to the existence of ``double fading" effect on the reflecting channel link between transmitter and user,  RIS helps achieve limited secrecy performance gain compared with the case without RIS. In this correspondence, we propose a novel active RIS design to enhance the secure wireless transmission, where the reflecting elements in RIS  not only adjust the phase shift but also amplify the amplitude of signals. To solve the non-convex secrecy rate optimization based on this design, an efficient alternating optimization algorithm is proposed to jointly optimize the beamformer at transmitter and  reflecting coefficient matrix at RIS. Simulation results show that with the aid of active RIS design, the impact of ``double fading" effect can be effectively relieved, resulting in a significantly higher secrecy performance gain compared with existing solutions with passive RIS  and without RIS design.

\end{abstract}	


\section{Introduction}

Reconfigurable intelligent surface (RIS),  is a  software-controlled metasurface equipped with low complexity passive reflecting elements. These elements are able to  induce certain phase shift  for the incident electromagnetic signal waves. And with proper phase shifts,  the quality of communications to user can be greatly enhanced, which is a very promising  technique for beyond 5G and 6G communications \cite{You-21}.

 Recently, RIS has been used in physical layer security to enhance the secure transmission in multi-antenna systems. 
By jointly optimizing the beamforming vector at transmitter and  phase shifts at RIS, the signals transmitted via the direct and reflecting channel link can be constructively (destructively) added at user (eavesdropper), resulting in a higher secrecy rate (SR) than the case without RIS. Secure RIS-assisted wireless communications in multi-input single-output (MISO) channel was  investigated first, and numerical algorithm was proposed to enhance the secrecy rate under  full channel state information (CSI) \cite{Cui-19}-\cite{Guan-19}.  In \cite{Dong-20d},  artificial noise (AN) approach is proposed to enhance the SR based on full and unknown CSI of eavesdropper. Later on, RIS was also  applied to  enhance the SR in cognitive radio systems, and different numerical solutions were proposed  under different cases of CSI condition \cite{Dong-21}. A double RIS assisted MISO system was studied in \cite{Dong-21b} and an iterative product Riemannian manifold algorithm is developed to optimize the phase shift of RIS.  Furthermore, RIS was also used to enhance the secure transmission in multi-input multi-output systems, and some efficient algorithms were proposed to maximize the SR \cite{Dong-20c}\cite{Dong-20} as well as AN-aided methods \cite{Hong-20}\cite{Chu-21}. All these works indicate that   RIS has effectively improved the SR compared with the existing solutions without RIS. 

However, all these aforementioned existing works have ignored an unavoidable problem  in RIS-assisted system: although RIS brings new reliable reflection link for signal transmission in addition to the direct link, a ``double fading" effect always exists in this reflection link, i.e., the signals received via this link suffer from  large-scale fading twice \cite{Vincent-21}. And if  the fading coefficient is  large, the signals from this longer reflection link lose more power than that from the short direct link, resulting in a limited performance gain compared with the one without RIS. To combat  this  ``double fading" effect, recently, a new  concept of active RIS has been proposed \cite{Vincent-21}-\cite{Xu-21}. The key feature of active RIS is that each reflecting element is equipped with a power amplifier so that the phase and the magnification of the signal can be simultaneously adjusted at the cost of  additional power supply. In \cite{Vincent-21} and \cite{Liang-21}\cite{Xu-21}, an active RIS assisted single-input multiple-output uplink system  and downlink MISO multi-user system were studied respectively, and numerical algorithms were developed to optimize the phase shift and amplification factor of active RIS. Based on the numerical examples, this active RIS design can realize a significantly higher transmission rate compared with those via passive RIS, thus overcoming the fundamental limit of “double fading” effect. 

However, all these current contributions \cite{Vincent-21}-\cite{Xu-21} of active RIS study  are restricted to non-secure setting, i.e., no eavesdropper is considered in the system model. Although active RIS greatly helps enhancing the quality of communications of user, it also reduces “double fading” effect in the reflection link of base station-RIS-eavesdropper due to the broadcast nature of wireless channels, resulting larger information leakage to eavesdropper. Hence, the current proposed algorithms  is not applicable for enhancing the SR in secure communication case, and the transmit coefficients at base station and RIS should be carefully re-designed via new efficient algorithms.

Against these aforementioned backgrounds, in this correspondence,  we propose a novel active RIS-assisted design to enhance the secure wireless transmission in multi-antenna systems. To the best of our knowledge, this is the first work that applies  active RIS  in physical layer security. Specifically, we design a new signal model based on this active RIS assisted system and formulate the SR optimization problem. To solve this non-convex problem, we develop an efficient alternating optimization  (AO) algorithm to jointly optimize the  beamforming vector at  transmitter and reflecting matrix at  RIS. Simulation results show that  our proposed algorithm with active RIS design  not only effectively relieve the “double fading” effect of the reflecting link from transmitter to legitimate user, and  also boost the SR  compared with the existing solutions for passive RIS and no RIS design.

\emph{Notations}:
 $\bA^{T}$, $\bA^{*}$ and $\bA^{H}$ denote transpose, conjugate and  Hermitian conjugate of $\bA$ respectively;    $tr(\bA)$ is the trace of $\bA$;  $\bu_{max}(\bA)$ denotes the eigenvector corresponding to the largest eigenvalue of $\bA$;  $|\ba|$ denotes the Euclidean norm of $\ba$; $|\bA|_F$ denotes the F-norm of $\bA$; $\bA[i,j]$ denotes the entry in the $i$-th row and $j$-th column; $\ba[i]$ denotes the $i$-th entry in $\ba$; $arg(a)$ is the phase of complex variable $a$; $\bu_{max}(\bA)$ denotes the eigenvector corresponding to the largest eigenvalue of $\bA$, and $\lambda_{max}(\bA)$ is the largest eigenvalue of $\bA$; $\ba[1:n]$ is to extract the first $n$ entires in $\ba$; $diag(\ba)$ is to transform $\ba$ to a diagonal matrix with diagonal elements in $\ba$.

\section{Channel Model And Problem Formulation}

Let us consider an RIS-assisted multi-antenna wiretap system shown in Fig.1, in which an RIS,  a transmitter (Alice), receiver (Bob) as
well as an eavesdropper (Eve) are included. Alice is equipped with $m$ antennas, RIS is equipped with $n$ reflecting
elements, and both Bob and Eve are equipped with a single antenna respectively. Identically with the settings in \cite{Vincent-21}-\cite{Xu-21}, the RIS is in active mode and is powered by an independent energy source. Furthermore, we remark that the proposed active RIS assisted system is significantly different from the traditional amplify-and-forward (AF) relay-assisted system. The active RIS is not equipped with radio frequency (RF) chains, which are used for baseband signal processing. Hence, identical  with passive RIS, the active RIS does not have the ability for signal storage but only reflects the incident signals so that the signals from both the  reflection and direct links  are received at user simultaneously in a single time phase. While for AF relay system, it generally needs power-consuming RF chains  to store the signal first and then transmit it to the user. Therefore, the signals transmitted via the link of base station-user and  base station-relay-user are received at user in two different time phases. Hence, the signal model between the active RIS and relay assisted system are also different, which cannot apply the same numerical solution to maximize the SR. Note that this key diffierence also has been mentioned in the current literature for the study of active RIS assisted non-secure system \cite{Vincent-21}-\cite{Xu-21}.

Based on this setting, the signals received at Bob $y_{B}$ and at Eve $y_{E}$ are expressed as
\bal
\notag
&y_B=(\bh_{AB}+\bh_{IB}\bQ\bH_{AI})\bx+\bh_{IB}\bQ\bn_I+n_B,\\
\notag
&y_E=(\bh_{AE}+\bh_{IE}\bQ\bH_{AI})\bx+\bh_{IE}\bQ\bn_I+n_E,
\eal
respectively, where $\bh_{AB}\in\mathbb{C}^{1\times m}$, $\bh_{IB}\in\mathbb{C}^{1\times m}$, $\bH_{AI}\in\mathbb{C}^{n\times m}$, $\bh_{AE}\in\mathbb{C}^{1\times m}$, $\bh_{IE}\in\mathbb{C}^{1\times n}$ denote the channel links of Alice-Bob, RIS-Bob, Alice-RIS, Alice-Eve, RIS-Eve respectively, $\bx=\bw s$ is the transmitted signal, $s$ is the confidential message following the distribution of zero mean and unit variance, $\bw$ is the beamformer, $\bQ\in\mathbb{C}^{n\times n}$ is the diagonal reflecting coefficient matrix at RIS, in which  $|\bQ[i,i]|$ and $arg(\bQ[i,i])$ represent the amplification factor and phase shift coefficient respectively at the $i$-th reflecting element,  $n_{B}\sim\mathcal{CN}(0,\sigma ^2_B)$ and $n_{E}\sim\mathcal{CN}(0,\sigma ^2_E)$  represent complex noise at  Bob and Eve respectively, $\sigma ^2_B$, $\sigma ^2_E$ denote the noise power. 

In addition, we remark that since RIS is equipped with  amplifiers on the reflecting elements, the thermal noise $\bn_I\sim\mathcal{CN}(0,\sigma ^2_I\bI)$ generated at RIS cannot be ignored, which is significantly different from the passive RIS case \cite{Vincent-21}-\cite{Xu-21}. We assume that full channel state information (CSI) of all the channel links
are available at Alice,  this can be realized since Eve is just
another user in the system and it also share its CSI with Alice
but is untrusted by Bob \cite{Cui-19}-\cite{Dong-21b}. The study based on full CSI also can be served as a theoretical system performance benchmark.  For how to estimate the CSI of all the channel links, we consider a  quasi-static block fading channel, and focus on one particular fading block  over which all the channels remain approximately constant. Then,  we apply time division scheme  \cite{Wang-19} to estimate the direct and reflection link channels in 4 independent time slot. 
In the first and second time slot, we shut down the RIS and estimate the channel $\bh_{AB}$ and $\bh_{AE}$ respectively. In the third and fourth time slots,  the RIS is in operating mode but without phase shift and amplification. Then the reflecting channels $\bH_{AI}$ and $\bh_{IB}$ are estimated in the third time slot, and the channel $\bh_{IE}$ is estimated in the fourth time slot.

Therefore, for $i=1,2,...,n$,  the SR optimization problem for this channel model is formulated as
\bal
\notag
&(P1)\quad\quad \underset{\bw,  \bQ}{\max} \ \ln \left( 1+\frac{|(\bh_{AB}+\bh_{IB}\bQ\bH_{AI})\bw|^2}{\sigma_B^2+|\bh_{IB}\bQ|^2\sigma^2_I} \right)\\
\notag
  &\quad\quad\quad\quad\quad-\ln\left(1+\frac{|(\bh_{AE}+\bh_{IE}\bQ\bH_{AI})\bw|^2}{\sigma_E^2+|\bh_{IE}\bQ|^2\sigma^2_I}\right),\\
\notag
&s.t.\  |\bw|^2\leq P_T, |\bQ\bH_{AI}\bw|^2+|\bQ|^2_F\sigma^2_I\leq P_I, \quad |\bQ[i,i]|\leq \eta_i
\eal
where $P_T$ is the maximum transmit power budget at Alice, $P_I$ is the maximum amplification power budget at RIS, $\eta_i>1$ is the  maximum  amplification factor  at the $i$-th reflecting element. 

\begin{figure}[t]
	\centerline{\includegraphics[width=2.2in]{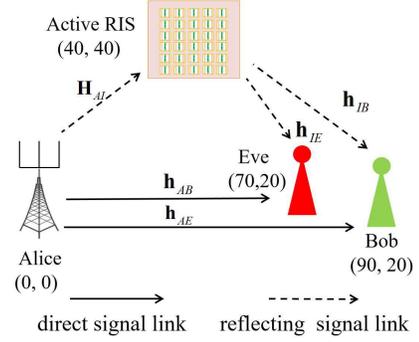}}
	\caption{An RIS-assisted multi-antenna wiretap system.}
\end{figure}

\begin{remark}
It is noted that for passive RIS, $\bQ$ should satisfy the unit modulus constraint $|\bQ[i,i]|=1$, i.e., the amplitude of the signal cannot be changed. But for the active RIS, the  amplification of the signal varies from 0 to $\eta$. Also note that the objective function becomes complicated and the difficulty for optimizing this problem increases due to the existence of the extra noise terms $|\bh_{IB}\bQ|^2\sigma_I^2$ and $|\bh_{IE}\bQ|^2\sigma_I^2$ in the objective function  compared with the simpler one for passive RIS.
\end{remark}

\section{AO Algorithm for SR Maximization}

To solve $P1$, in this section, we propose AO algorithm
to optimize $\bw$ and $\bQ$ in two sub-problems. When $\bQ$ is fixed,  semi-definite relaxation (SDR) in combination with Charnes-Cooper transformation (CCT) algorithm  is proposed to optimize $\bw$ globally. When $\bw$ is fixed, we apply SDR in combination with minorization-maximization (MM) algorithm to obtain a first-order optimal solution of $\bQ$. Furthermore, to recover the optimal rank-1 solution in the SDR, we propose a penalty based  algorithm to achieve the rank-1 solution globally. Finally, as the convergence of AO algorithm is reached, a Karush-Kuhn-Tucker (KKT) solution of $P1$ can be obtained.

\subsection{Algorithm for optimizing $\bw$ given $\bQ$}

Firstly, when $\bQ$ is fixed, the sub-problem of optimizing $\bw$ can be expressed as
\bal
\notag
(P2)\quad\quad  &\underset{\bw}{\max}\ (1+|\tilde{\bh}_B\bw|^2)(1+|\tilde{\bh}_E\bw|^2)^{-1}\\
\notag
&s.t.\quad|\bw|^2\leq P_T, |\bQ\bH_{AI}\bw|^2\leq \tilde{P}_I
\eal
where the log is omitted due to its monotonicity, $\tilde{P}_I=P_I-|\bQ|^2_F\sigma^2_I$  and where
 $\tilde{\bh}_B=\bh_B/\sqrt{\sigma_B^2+|\bh_{IB}\bQ|^2\sigma^2_I}, 
\tilde{\bh}_E=\bh_E/\sqrt{\sigma_E^2+|\bh_{IE}\bQ|^2\sigma^2_I},
\bh_B=\bh_{AB}+\bh_{IB}\bQ\bH_{AI}$, and $\bh_E=\bh_{AE}+\bh_{IE}\bQ\bH_{AI}$.
 To solve this non-convex problem, we use the key idea of SDR technique and transform the problem to a relaxed one
\bal
\notag
(P3)\quad  &\underset{\bS}{\max}\ (1+\tilde{\bh}_B\bS\tilde{\bh}_B^H)(1+\tilde{\bh}_E\bS\tilde{\bh}_E^H)^{-1}\\
\notag
&s.t.\ \bS\geq\bo, tr(\bS)\leq P_T, tr(\bQ\bH_{AI}\bS\bH_{AI}^H\bQ^H)\leq \tilde{P}_I
\eal
where $\bS=\bw\bw^H$ and  the constraint $rank(\bS)=1$ is omitted. $P3$ is a quasiconcave problem with convex constraints. Let $\tilde{\bS}=t\bS$,  and $t=1/(1+\tilde{\bh}_E\bS\tilde{\bh}_E^H)$. By applying CCT method, $P3$ can be  equivalently transformed to
\bal
\notag
&(P4)\quad  \underset{\tilde{\bS}, t}{\max}\ t+\tilde{\bh}_B\tilde{\bS}\tilde{\bh}_B^H\\
\label{eq1}
 &s.t.\quad tr(\tilde{\bS})\leq tP_T, tr(\bQ\bH_{AI}\tilde{\bS}\bH_{AI}^H\bQ^H)\leq t\tilde{P}_I,\\
\label{eq2}
 &\quad\quad\quad\tilde{\bh}_E\tilde{\bS}\tilde{\bh}_E^H+t=1, t\geq0, \tilde{\bS}\geq\bo,
\eal 
which can be directly optimized via CVX solver. Since the rank-1
constraint of $\tilde{\bS}$ is omitted here, the solution may not be of
rank-1 so that the original $\bw$ cannot be recovered. Therefore, we propose a penalty based approach to recover the optimal rank-1 solution. Specifically, the constraint $rank(\tilde{\bS})=1$ can be equivalently expressed as $tr(\tilde{\bS})-\lambda_{max}(\tilde{\bS})\leq 0$, and this holds for any Hermitian $\tilde{\bS}$. Then we construct the following problem 
\bal
\notag
(P5)\ \underset{\tilde{\bS}, t}{\min}\ -(t+\tilde{\bh}_B\tilde{\bS}\tilde{\bh}_B^H)+p(tr(\tilde{\bS})-\lambda_{max}(\tilde{\bS}))\ s.t. \eqref{eq1}, \eqref{eq2}
\eal
where $p>0$ is a penalty weight with large enough value to achieve the small value of $tr(\tilde{\bS})-\lambda_{max}(\tilde{\bS})$. Using the key idea of proof in \cite{Phan-12}, one obtains that there must exist $0<p_0<+\infty$ such that for any $p>p_0$, $P4$ and $P5$ share the same optimal solution as well as the optimal value. Since the objective function in $P5$ is concave, $P5$ is a class of concave programming problem. Furthermore, $\lambda_{max}(\tilde{\bS})$ is not differentiable, so we can apply the subgradient of $\lambda_{max}(\tilde{\bS})$ as $\bu_{max}(\tilde{\bS})\bu_{max}(\tilde{\bS})^H$. Hence, given a feasible solution $\tilde{\bS}^{(k)}$ for $P5$ in the $k$-th iteration, the following problem gives an improved solution of $P5$:
\bal
\notag
&(P6)\quad \underset{\tilde{\bS}, t}{\min}\ -(t+\tilde{\bh}_B\tilde{\bS}\tilde{\bh}_B^H)\\
\notag
&+p(tr(\tilde{\bS})-\bu_{max}(\tilde{\bS}^{(k)})^H\tilde{\bS}\bu_{max}(\tilde{\bS}^{(k)}))\quad s.t.\ \eqref{eq1}, \eqref{eq2}.
\eal
We summarize the penalty based algorithm for recovering rank-1 solution of $P4$ as Algorithm 1. In this algorithm, $|\tilde{\bS}^{(k+1)}-\tilde{\bS}^{(k)}|_F\leq\epsilon_2$  means that there is no improved solution, then $p$ will be updated. Once an improved solution is found, $\tilde{\bS}^{(k+1)}$ will be set as the new initial point for the next iteration. If $|tr(\tilde{\bS}^{(k)})-\lambda_{max}(\tilde{\bS}^{(k)})|$ is below the target accuracy $\epsilon_1$, then  the algorithm terminates and the optimal rank-1 of $\bS$ is obtained.

\begin{algorithm}[h]
	\caption{(\it Penalty based algorithm for recovering rank-1 solution in $P4$)}
	\begin{algorithmic}
    \Require  $\epsilon_1, \epsilon_2>0$,  $p=10$,  feasible point $\tilde{\bS}^{(0)}$, $t^{(0)}$ for $P4$.
        \State 1. Set $k=0$. 
\Repeat \  \ \ 
           \State 2. If $|tr(\tilde{\bS}^{(k)})-\lambda_{max}(\tilde{\bS}^{(k)})|\leq\epsilon_1$, stop algorithm and go to step 5, else go to step 3.
           \State 3. Optimize $\tilde{\bS}^{(k+1)}$, $t^{(k+1)}$ in $P6$ given $\bu_{max}(\tilde{\bS}^{(k)})$.                                                                                                                                                                                                                        
     \State 4. If $|\tilde{\bS}^{(k+1)}-\tilde{\bS}^{(k)}|_F\leq\epsilon_2$, set $p=2p$, go to step 3, else set $k=k+1$, $\tilde{\bS}^{(k)}=\tilde{\bS}^{(k+1)}$ go to step 2.
 \Until  $|tr(\tilde{\bS}^{(k)})-\lambda_{max}(\tilde{\bS}^{(k)})|\leq\epsilon_1$
\State 5.  Output $\bS^{(k)}=\tilde{\bS}^{(k)}/t^{(k)}$ as the optimal rank-1 solution for $P3$.
	\end{algorithmic}
\end{algorithm}

\begin{prop}
Denote $f(t,\tilde{\bS})$ as the objective function of $P5$, and consider that $(t^{(k)},\tilde{\bS}^{(k)})$ is the optimal solution of $P6$ given  fixed $p$,  $f(t^{(k+1)},\tilde{\bS}^{(k+1)})\leq f(t^{(k)},\tilde{\bS}^{(k)})$, i.e., the convergence of the sequence $f(t^{(k)},\tilde{\bS}^{(k)})$, $k=0,1,...$, is non-increasing.
\end{prop}

\begin{proof}
Firstly, note that $P6$ is also equivalent to the following problem:
\bal
\notag
&(P6') \quad \underset{\tilde{\bS}, t}{\min}\ -(t+\tilde{\bh}_B\tilde{\bS}\tilde{\bh}_B^H)+p(tr(\tilde{\bS})-\lambda_{max}(\tilde{\bS}^{(k)})\\
\notag
&-\bu_{max}(\tilde{\bS}^{(k)})^H(\tilde{\bS}-\tilde{\bS}^{(k)})\bu_{max}(\tilde{\bS}^{(k)}))\ s.t.\ \eqref{eq1}, \eqref{eq2}.
\eal
Then, consider that $(t^{(k+1)},\tilde{\bS}^{(k+1)})$ is the optimal solution of $P6'$, since $(t^{(k)},\tilde{\bS}^{(k)})$ is also feasible for $P6'$, one obtains that
$-(t^{(k+1)}+\tilde{\bh}_B\tilde{\bS}^{(k+1)}\tilde{\bh}_B^H)+p(tr(\tilde{\bS}^{(k+1)})-\lambda_{max}(\tilde{\bS}^{(k)})-\bu_{max}(\tilde{\bS}^{(k)})^H(\tilde{\bS}^{(k+1)}-\tilde{\bS}^{(k)})\bu_{max}(\tilde{\bS}^{(k)}))\leq-(t^{(k)}+\tilde{\bh}_B\tilde{\bS}^{(k)}\tilde{\bh}_B^H)+p(tr(\tilde{\bS}^{(k)})-\lambda_{max}(\tilde{\bS}^{(k)}))$.
Secondly,  the sub-gradient of $\lambda(\tilde{\bS}^{(k)})$ is $\partial\lambda(\tilde{\bS}^{(k)})=\bu_{max}(\tilde{\bS}^{(k)})\bu_{max}(\tilde{\bS}^{(k)})^H$, and $\forall \bY\geq\bo$, $\lambda_{max}(\bY)-\lambda_{max}(\bX)\geq  \bu_{max}(\tilde{\bS}^{(k)})^H(\bY-\bX)\bu_{max}(\tilde{\bS}^{(k)})$ holds \cite{J-96}. Therefore,  we have
\bal
\notag
&f(t^{(k+1)},\tilde{\bS}^{(k+1)})
\\
\notag
=&-(t^{(k+1)}+\tilde{\bh}_B\tilde{\bS}^{(k+1)}\tilde{\bh}_B^H)+p(tr(\tilde{\bS}^{(k+1)})-\lambda_{max}(\tilde{\bS}^{(k+1)})\\
\notag
&-\lambda_{max}(\tilde{\bS}^{(k)})+\lambda_{max}(\tilde{\bS}^{(k)}))\\
\notag
\leq&-(t^{(k+1)}+\tilde{\bh}_B\tilde{\bS}^{(k+1)}\tilde{\bh}_B^H)+p(tr(\tilde{\bS}^{(k+1)})-\lambda_{max}(\tilde{\bS}^{(k)})\\
\notag
&-\bu_{max}(\tilde{\bS}^{(k)})^H(\tilde{\bS}^{(k+1)}-\tilde{\bS}^{(k)})\bu_{max}(\tilde{\bS}^{(k)}))\\
\notag
\leq&-(t^{(k)}+\tilde{\bh}_B\tilde{\bS}^{(k)}\tilde{\bh}_B^H)+p(tr(\tilde{\bS}^{(k)})-\lambda_{max}(\tilde{\bS}^{(k)}))\\
\notag
=&f(t^{(k)},\tilde{\bS}^{(k)}),
\eal
from which the proof is complete.
\end{proof}

Proposition 1 indicates that Algorithm 1 is guaranteed to achieve a monotonic convergence, and after convergence, an optimal rank-1 solution of $P4$ can be achieved \cite{Phan-12}. Once the rank-1 $\bS$ is obtained, then the optimal beamformer of $P2$ can be expressed as $\bw=\bu_{max}(\bS)\sqrt{\lambda_{max}(\bS)}$.

\subsection{Algorithm for optimizing $\bQ$ given $\bw$}

The next step is to optimize $\bQ$ at RIS given $\bw$. To make the problem more tractable, let $\bv=\bq^*\in\mathbb{C}^{n\times 1}$,  where the entries in $\bq$ are all the diagonal elements in $\bQ$, then the  sub-problem of optimizing $\bQ$ is equivalently expressed as
\bal
\notag
&(P7)\quad\quad  \underset{\bv}{\max}\ C_B-C_E\\
\notag
&s.t.\  |\bv^Hdiag(\bH_{AI}\bw)|^2+|\bv|^2\sigma^2_I\leq P_I, \quad |\bv[i]|\leq \eta_i
\eal
where for $j\in\{B,E\}$, 
$C_j=\ln(\sigma^2_j+|\bv^Hdiag(\bh_{Ij})|^2\sigma^2_I+|(\bh_{Aj}+\bv^H\bH_j)\bw|^2)-\ln(\sigma^2_j+|\bv^Hdiag(\bh_{Ij})|^2\sigma^2_I), \bH_j=diag(\bh_{Ij})\bH_{AI}$. 
To solve this non-convex problem, SDR method is applied again. Letting
\bal
\notag
 \bV=\begin{bmatrix}
\bv\\ 1
\end{bmatrix}\begin{bmatrix}
\bv^H&1 
\end{bmatrix}, 
\eal
and $P7$ is relaxed as 
\bal
\notag
&(P8)\quad \underset{\bV}{\max}\  C(\bV)=\bar{C}_B-\bar{C}_E\\ 
\label{eq4}
&s.t.\ tr(\bH_A\bV)\leq P_I, \bV[i,i]\leq\eta_i^2,\forall i,\\
\label{eq5}
&\quad\quad \bV[n+1,n+1]=1, \bV\geq\bo,
\eal
where the rank-1 constraint is omitted, $\bar{C}_j=\ln(tr(\bH_{Aj}\bV))-\ln(tr(\bH_{Ij}\bV))$, $\tau_j=\sigma_j^2+\bh_{Aj}\bw\bw^H\bh_{Aj}^H$ and 
\bal
\notag
&\bH_A=\begin{bmatrix}
        diag(\bH_{AI}\bw)diag(\bH_{AI}\bw)^H+\sigma_I^2\bI   &  \bo_{n\times 1}  \\ 
       \bo_{1\times n}    &  0
\end{bmatrix},\\
\notag
&\bH_{Aj}=\begin{bmatrix}
\bH_j\bw\bw^H\bH_j^H+\bar{\bH}_j     & \bH_j\bw\bw^H\bh_{Aj}^H\\ 
 \bh_{Aj}\bw\bw^H\bH_j^H& \tau_j
\end{bmatrix},\\
\notag
&\bH_{Ij}= \begin{bmatrix}
\bar{\bH}_j & \bo_{n\times 1}\\ 
 \bo_{1\times n}& \sigma_j^2
\end{bmatrix}, \bar{\bH}_j=\sigma_I^2diag(\bh_{Ij})diag(\bh_{Ij})^H.
\eal
$P8$ is still non-convex due to the complicated non-convex objective function. To solve this problem, our key idea is to firstly relax $P8$ by approximating the objective function, and then propose MM algorithm to iteratively optimize the relaxed problem. During each iteration of MM, we apply the aforementioned penalty based Algorithm 1 again to recover the rank-1 solution.

\begin{prop}
Given fixed feasible point $\tilde{\bV}$, the  function $C(\bV)$ in $P8$ can be lower bounded as
\bal
\notag
C(\bV)\geq&\ln(tr(\bH_{AB}\bV))+\ln(tr(\bH_{IE}\bV))-\ln(tr(\bH_{IB}\tilde{\bV}))\\
\notag
&-tr(\bH_{IB}/tr(\bH_{IB}\tilde{\bV})(\bV-\tilde{\bV}))-\ln(tr(\bH_{AE}\tilde{\bV}))\\
\notag
&-tr(\bH_{AE}/tr(\bH_{AE}\tilde{\bV})(\bV-\tilde{\bV}))=\tilde{C}(\bV;\tilde{\bV}),                
\eal

and $\tilde{C}(\bV;\tilde{\bV})$ is a surrogate function.
\end{prop}

\begin{proof}
Proposition 2 follows from the first-order Taylor expansion theorem \cite{Wang-17}\cite{Sun-17}: for any concave function $f(\bx)$ with $\bx, \tilde{\bx} \in \bf dom $ $f$,
$f(\bx)\leq f(\tilde{\bx})+(\nabla f(\tilde{\bx}))^T(\bx-\tilde{\bx})$, 
and for any concave function $g(\bX)$ with $\bX, \tilde{\bX} \in \bf dom $ $g$,
$g(\bX)\leq g(\tilde{\bX})+tr(\nabla g(\tilde{\bX})(\bX-\tilde{\bX}))$. 
Hence, since $tr(\bH_{IB}\bV)$ and $tr(\bH_{AE}\bV)$ are both linear respect to $\bV$, $\ln(tr(\bH_{IB}\bV))$ and $\ln(tr(\bH_{AE}\bV))$ are both concave respect to $\bV$.
Then, given feasible $\tilde{\bV}$, $\ln(tr(\bH_{IB}\bV))\leq \ln(tr(\bH_{IB}\tilde{\bV}))+tr(\bH_{IB}/tr(\bH_{IB}\tilde{\bV})(\bV-\tilde{\bV}))$ and $\ln(tr(\bH_{AE}\bV))\leq \ln(tr(\bH_{AE}\tilde{\bV}))+tr(\bH_{AE}/tr(\bH_{AE}\tilde{\bV})(\bV-\tilde{\bV}))$, from which $C(\bV)\leq\tilde{C}(\bV;\tilde{\bV})$ holds.

$\tilde{C}(\bV;\tilde{\bV})$ is a surrogate function since four key conditions holds \cite{Sun-17}: 1).   $C(\bV)\geq\tilde{C}(\bV;\tilde{\bV})$ holds; 2). $C(\tilde{\bV})=\tilde{C}(\tilde{\bV};\tilde{\bV})$; 3). $\nabla\tilde{C}(\bV;\tilde{\bV})|_{\bV=\tilde{\bV}}=\nabla C(\bV)|_{\bV=\tilde{\bV}}$; 4). $\tilde{C}(\bV;\tilde{\bV})$ are both continuous in $\bV$ and $\tilde{\bV}$.
\end{proof}

Therefore, after dropping the constant terms in $\tilde{C}(\bV;\tilde{\bV})$, $P8$ can be approximated as 
\bal
\notag
&(P9)\quad  \underset{\bV}{\max}\ \ln(tr(\bH_{AB}\bV))+\ln(tr(\bH_{IE}\bV))\\
\notag
&-tr((\bH_{IB}/tr(\bH_{IB}\tilde{\bV})+\bH_{AE}/tr(\bH_{AE}\tilde{\bV}))\bV)\ s.t.\ \eqref{eq4}, \eqref{eq5},
\eal
which is convex and can be directly solved via CVX. And  the penalty based algorithm shown above can be applied again to recover the rank-1 solution if $\bV$ is of not rank-1. Once $\bV$ is optimized, it is denoted as the new initial point $\tilde{\bV}$  and $\bV$ is optimized again in $P9$. According to the key property of MM,  the convergence is non-decreasing and a first-order optimal KKT point for $P8$ can be obtained. After obtaining $\bV$, the original precoding matrix can be reformulated as $\bQ=diag((\bu_{max}(\bV)\sqrt{\lambda_{max}(\bV)})[1:n])$. The SDR+MM algorithm  for optimizing $\bV$ is summarized in Algorithm 2, in which $C^{(k)}(\bV)$ denotes the value of objective function in $P8$ in the $k$-th iteration.

\begin{algorithm}[h]
	\caption{(\it SDR+MM based algorithm for optimizing $\bV$)}
	\begin{algorithmic}
    \Require  $\epsilon_3>0$, feasible point $\tilde{\bV}$.
        \State 1. Set $k=0$, compute $C^{(k)}(\bV)$.
\Repeat \  \ \ 
           \State 2. Optimize $\bV$ in $P9$ given $\tilde{\bV}$.
         \State 3. If $rank(\bV)>1$, recover the rank-1 solution via penalty based method.
         \State 4. Obtain the solution $\bV$, Compute $C^{(k+1)}(\bV)$.
\State 5. If $|C^{(k+1)}(\bV)-C^{(k)}(\bV)|\leq\epsilon_3$, stop algorithm, else set $k=k+1$, $\tilde{\bV}=\bV$, go back to step 2.
 \Until  $|C^{(k+1)}(\bV)-C^{(k)}(\bV)|\leq\epsilon_3$
\State 6.  Output $\bV$ as the first-order optimal KKT point of  $P8$.
	\end{algorithmic}
\end{algorithm}

\subsection{Convergence and complexity of the  algorithm}

Finally, since $\bw$ and $\bQ$ are both bounded by the constraints, according to the  Cauchy’s theorem \cite{Dong-20c}, the AO algorithm is guaranteed to have a monotonic convergence. As the convergence is reached, a KKT point solution of $\bw$ and $\bQ$ for the original problem $P1$ can be obtained. In the AO algorithm, the main computational complexity for optimizing $\bw$ given $\bQ$ and optimizing  $\bQ$ given $\bw$ are about $\mathcal{O}(n^{3})$ and $\mathcal{O}((n+1)^{3})$ respectively.

\section{Simulation Results}

To validate the SR performance achieved by the proposed algorithm with the active RIS design, numerical simulations have been carried out in this section.  Following \cite{Cui-19}\cite{Hong-20}, we consider that all the channels are formulated as the product of large scale fading and small scale fading. For the direct channels $\bh_{AB}$ and $\bh_{AE}$, the small scale fading is assumed to be Rayleigh fading due to extensive scatters. And the small scale fading in the reflecting channels $\bH_{AI}$, $\bh_{IB}$, $\bh_{IE}$ are assumed to be Rician fading. Taking $\bh_{AB}$ and $\bH_{AI}$ as examples, they are modeled as $\bh_{AB}=\bg_{AB}PL_{AB}$ and
\bal
\notag
\bH_{AI}=\left(\sqrt{\frac{\kappa}{\kappa+1}}\bar{\bH}_{AI}+\sqrt{\frac{1}{\kappa+1}}\tilde{\bH}_{AI}\right)PL_{AI},
\eal
where $\bg_{AB}$ and $\tilde{\bH}_{AI}$ are  Rayleigh fading channels, the entry of which is randomly generated  complex Gaussian random variables with zero mean, unit
variance, $PL_{AB}=\sqrt{\beta(d_0/d_{AB})^{\alpha_{AB}}}$ and $PL_{AI}=\sqrt{\beta(d_0/d_{AI})^{\alpha_{AI}}}$, $\beta=-30$dB is the pathloss at reference distance $d_0=$1m, $d_{AB}$ and $d_{AI}$ are the distance of Alice-Bob and Alice-RIS respectively, $\kappa=5$ is the Rician factor,  $\bar{\bH}_{AI}= \boldsymbol{\alpha}_I(\phi_{AI})\boldsymbol{\alpha}_A^H(\theta_{AI})$,
\bal
\notag
&\boldsymbol{\alpha}_I(\phi_{AI})=\begin{bmatrix}
1 & e^{j \frac{2\pi d_r cos(\phi_{AI})}{\lambda}} & ... & e^{j \frac{2(n-1)\pi d_r cos(\phi_{AI}) }{\lambda}}
\end{bmatrix}^T,\\
\notag
&\boldsymbol{\alpha}_A(\theta_{AI})=\begin{bmatrix}
1 & e^{j \frac{2\pi d_t cos(\theta_{AI})}{\lambda}} & ... & e^{j \frac{2(n-1)\pi d_t cos(\theta_{AI}) }{\lambda}}
\end{bmatrix}^T,
\eal
and where  $\theta_{AI}=arccos(x_I/d_{AI})$ and $\phi_{AI}=\pi-\theta_{AI}$ denote the 
angle of departure and the angle of arrival respectively, $x_I$ is the coordinate of RIS in the x-axis, $\lambda$  is the wavelength, $d_t$ and $d_r$ are the element intervals of the transmit and receive array. We set $d_t/\lambda=d_r/\lambda=0.5$, and let $\alpha_{AB}=3.8, \alpha_{AE}=3.5, \alpha_{AI}=\alpha_{IB}=\alpha_{IE}=2.2$ for each channel link, and the  noise power is $\sigma^2_B=\sigma^2_E=\sigma^2_I=-95$dBm.  For simplicity, we set that each amplification factor on the reflectin element is same, i.e., $\eta_i=\eta, \forall i$. For the location of each node, we consider a two dimensional coordinate space, and let Alice, Bob, RIS, Eve to be fixed in the coordinate $(0,0), (90,20),(40,40), (70,20)$ respectively (see Fig.1).
 In the penalty based algorithm, MM algorithm as well as AO algorithm, all the target accuracy is set as $10^{-3}$.  In the AO algorithm, the starting point is set as $\bw=\bo$, $\bQ=\bI$. All the results plotted in Fig.2 and Fig.3 are averaged over 100 channel realizations.

Fig.2 shows the SR performance versus $P_T$ achieved by the proposed algorithm with active RIS  as well as existing solution with passive RIS \cite{Cui-19} and optimal solution without RIS. Following \cite{Liang-21}, we set the amplification ability of each reflecting element in active RIS as  20, 30 or 40 in dB. Compared with the result by optimal solution without RIS, observe that due to ``double fading" effect, the SR achieved by passive RIS has only increased by about $3\%$, but the proposed AO algorithm with active RIS design achieves much better SR performance (see e.g., about $13\%$ performance gain when $\eta^2=$20dB, and about $48\%$ performance gain when $\eta^2$ increases to 40dB ). This indicates that active RIS design is effective on weakening the influence brought by ``double fading" effect in the reflecting channel link, and our proposed algorithm also guarantees a low information leakage at Eve, resulting in a higher SR than the existing solution with passive RIS. In fact, we remark that  although this great performance gain is at the cost of extra power consumption at RIS, only a very small fraction of the total reflect power $P_I$ is used for amplification. In other words,   in $P1$, the constraint $|\bQ\bH_{AI}\bw|^2+|\bQ|^2_F\sigma^2_I\leq P_I$ is always  inactive but $|\bQ[i,i]|\leq \eta_i$ is always active, which indicates that the setting of $\eta$ is the dominant parameter for boosting the SR. For example, based on our extensive tests, when $P_T$=40dBm, and $\eta^2$=20dB, the total power used for signal amplification at RIS is only about 4dBm, but a  $13\%$ performance gain can be achieved compared with passive RIS case. Hence, the proposed active RIS design also achieves a better energy efficiency than the passive RIS case.

\begin{figure}[t]
	\centerline{\includegraphics[width=3.0in]{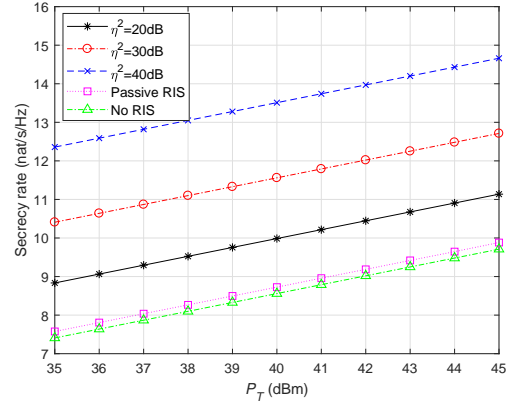}}
	\caption{SR performance versus $P_T$ under different solutions, $m=4, n=10$, $P_I=$40dBm.}
\end{figure}


Fig.3 illustrates the SR performance versus the number of reflecting elements $n$ under different solutions given fixed $P_T$ and $P_I$. It is noted that the SR returned by both active and passive RIS case increases with $n$ due to increased degree of freedom, but our proposed algorithm with active RIS design significantly performs better SR than the existing solutions with passive RIS and without RIS given fixed amplification factor. Furthermore, note that due to the ``double fading" effect, the SR with passive RIS only increases by about $15\%$ when $n$ varies from 10 to 60, which is less than $36\%$ achieved by active RIS when $\eta^2=$20dB. Even when $n=60$, the passive RIS only helps boosting the SR to about 8.6, which is close but still less than 9.0 achieved by active RIS when $\eta^2=$20dB and $n=10$. These results fully indicate that using active RIS can we save more number of reflecting elements to achieve a better performance gain compared with passive RIS case, thus greatly reducing the complexity of RIS.

\begin{figure}[t]
	\centerline{\includegraphics[width=3.0in]{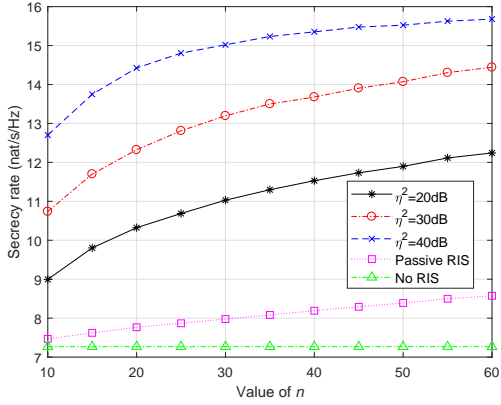}}
	\caption{SR performance versus $n$ achieved by different solutions, $m=4$, $P_T=40$dBm, $P_I=30$dBm.}
\end{figure}

Finally, Fig.4 gives the convergence of the proposed AO algorithm for optimizing $\bQ$ and $\bw$ as well as MM algorithm for
optimizing $\bQ$ given $\bw$ in one iteration of AO algorithm  based on different random channel
realizations. Note that for both AO and MM algorithms,
the convergence is monotonic under all settings of $m, n$. For
the MM  algorithm, it takes only 2 to 3 iterations to reach
target accuracy of $10^{-3}$. For AO algorithm, 4 to 8 iterations are required to converge. As the convergence is reached in both AO and MM algorithm, larger SR can be achieved given larger $m$ and $n$ due to the added  degree of freedom for transmission at Alice and reflecting elements with amplifiers  at RIS.

\begin{figure}[t]
	\centerline{\includegraphics[width=3.0in]{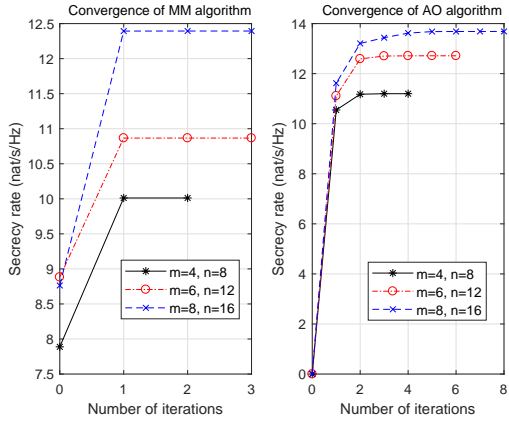}}
	\caption{Convergence of the AO and MM algorithms under different settings of $m$ and $n$, $P_T=P_I=30$dBm, $\eta^2=$30dB.}
\end{figure}

\section{Conclusion}

In this correspondence,  a novel active RIS-assisted secure wireless system is studied, in which the RIS helps enhancing the secure transmission by simultaneously adjusting the phase shift and amplitude of the signals. To solve the non-convex
SR optimization based on this design, an  AO algorithm is proposed to jointly optimize the beamformer at transmitter and reflecting coefficient matrix at RIS. Simulation results show that  the impact of ``double fading" effect can be effectively relieved with the aid of active RIS, the proposed algorithm greatly boosts the SR performance compared with existing algorithms with passive RIS and without RIS. 



\end{document}